\documentclass[12pt]{article}
\usepackage[dvips]{graphicx}
\usepackage{amsmath, amssymb}
\usepackage{amsthm}
\usepackage{bm}

\newtheorem{lemma}{Lemma}[section]

\newtheorem{proposition}{Proposition}[section]

\newcommand{\cK}{{\cal K}}
\newcommand{\cP}{{\cal P}}

\newcommand{\rnp}{\rho}    

\newlength{\IndentI}
\newlength{\IndentII}
\newlength{\IndentIII}
\newlength{\WidthI}
\newlength{\WidthII}
\newlength{\WidthIII}
\title{Multistep Bayesian strategy in coin-tossing games and its 
application to asset trading games in continuous time}

\author{
  Kei Takeuchi\\
  Emeritus, Graduate School of Economics\\
  University of Tokyo\\
  Masayuki Kumon\\
  Risk Analysis Research Center\\
  Institute of Statistical Mathematics\\
  and\\
  Akimichi Takemura\\
  Graduate School of Information Science and Technology\\
  University of Tokyo}
\date{February, 2008}

\begin{document}
\maketitle

\begin{abstract}
  We study multistep Bayesian betting strategies in coin-tossing
  games in the framework of game-theoretic probability of Shafer and
  Vovk (2001).  We show that by a countable mixture of these 
  strategies, a gambler or an investor 
  can exploit arbitrary patterns of deviations of nature's
  moves from 
  independent Bernoulli trials.  We then 
  apply our scheme to asset trading games in continuous time and
  derive the exponential growth rate of the investor's capital when the
  variation exponent of the asset price path deviates from two.
\end{abstract}


\section{Introduction}
\label{sec:intro}

The field of game-theoretic probability and finance established by
Shafer and Vovk \cite{shafer/vovk:2001} has been rapidly developing in
many directions.
The present authors have been contributing to this exciting new field
by focusing mainly on explicit strategies of the
gambler and the growth rate of his capital
(\cite{takeuchi:2004}, \cite{takeuchi:2004b}, \cite{kumon-takemura},
\cite{ktt1}, \cite{ktt2}, 
\cite{takemura-suzuki}, \cite{horikoshi-takemura}, \cite{tkt1}).
Following the terminology of 
Shafer and Vovk \cite{shafer/vovk:2001}, 
we  refer to the gambler 
as Skeptic 
or Investor (Section \ref{sec:preliminary asset})
and refer  to nature as 
Reality 
or Market (Section \ref{sec:preliminary asset}).

In this paper we extend the results of \cite{ktt1} and \cite{tkt1}
by considering multistep Bayesian strategies. 
In \cite{ktt1} we considered a class of Bayesian
strategies for Skeptic in coin-tossing games, strategies 
that were based only on the past average of Reality's moves.
We proved the important fact that if Skeptic uses a Bayesian strategy
and Reality violates the strong law of large numbers (SLLN), then the
exponential growth rate of the Skeptic's capital process is very
accurately described in terms of the Kullback divergence between the
average of Reality's moves when she violates SLLN and the average when
she observes SLLN.  
Furthermore in \cite{tkt1} we applied Bayesian strategies for
coin-tossing games to asset trading games in continuous time.
If we  discretize a continuous-time game by an equi-spaced grid in the
state space (i.e.\ the vertical axis), then the continuous time game can be 
approximated by an embedded coin-tossing game.  Thus the results of
discrete time coin-tossing games can be applied to continuous-time
games.  
In particular we gave a proof of ``$\sqrt{{\rm d}t}$-effect'', i.e., 
Investor  can force Market  to choose a price path with
variation exponent equal to two, within an arbitrary small constant.

More generally, discretization of continuous-time game in \cite{tkt1}
is based on the requirement that Investor choose a countable number
of discrete stopping times against a continuous path chosen by Market.
This approach allows us to formulate and study continuous-time games
in game-theoretic probability within the conventional theory of
analysis, whereas in the book by Shafer and Vovk continuous-time games
were formulated as limits of discrete time games using 
nonstandard analysis.  Vovk has taken up this formulation
and is currently rapidly developing it  in 
\cite{vovk-cont-1}, \cite{vovk-cont-2}, \cite{vovk-cont-3}.  By these 
works of Vovk it has 
now become clear that many measure-theoretic results on
continuous-time stochastic processes can be more directly derived in
the framework of game-theoretic probability.  It should be emphasized
that the game-theoretic approach is advantageous because no probabilistic
assumptions on the paths are imposed a priori.  Instead, a stochastic
behavior of Market results from the protocol of the game.  This has a
far-reaching conceptual implications for the emergence of probability.

In the coin-tossing games Reality may deviate from independent
Bernoulli trials in a subtle way without violating the strong law of
large numbers.  For example, Reality could choose a deterministic
sequence where heads and tails alternate.  Then SLLN holds for this
kind of paths and they can not be prevented by Bayesian strategies in 
\cite{ktt1}.  We could use  ``contrarian'' strategies 
(\cite{horikoshi-takemura}) based on the
past average of Reality's moves for preventing this kind of paths.
However a more natural approach is to model autocorrelations between
successive moves of Reality. Note that the deterministic sequence of
alternating heads and tails can be regarded as a sequence with
the first-order autocorrelation of $-1$.  These considerations
naturally lead to multistep Bayesian strategies of the present paper.
Then by a countable mixture of these strategies Skeptic can
detect and exploit arbitrary patterns of deviations
from the sequence of  independent Bernoulli trials.

Our results have close relations to the ones in many fields.  Universal
source coding has been extensively studied in information theory
(e.g.\ Han and Kobayashi \cite{han-kobayashi}) and the equivalence of
source coding and betting is discussed in Cover and Thomas
\cite{cover-thomas}.  
Various notions of randomness have been studied
from the viewpoint of Kolmogorov complexity (e.g.\ Li and Vit{\'a}nyi
\cite{li-vitanyi}, Lambalgen \cite{lambalgen})  
and there exists an extensive literature on algorithmic theory of
randomness.  See two forthcoming books on algorithmic randomness by  
Downey and Hirschfeldt \cite{downey-hirschfeldt} and by
Nies \cite{nies}.
We discuss these relations in Section \ref{sec:discussions}.
 
The organization of this paper is as follows.  In Section
\ref{sec:preliminary} we set up notations for the coin-tossing game
and give some preliminary results.  
In Section \ref{sec:patterns} we consider two types of Reality's moves
which suggest deviations from independent Bernoulli trials. The first
is a block type pattern  and the second is
a Markovian type pattern.  We construct Skeptic's Bayesian
strategies which can exploit these non-randomnesses.  We show that
it is asymptotically always advantageous to exploit higher order
patterns in Reality's moves.
In Section \ref{sec:preliminary asset} we consider asset trading games in
continuous time and investigate the consequences of
high-frequency block type strategies and Markovian strategies. We derive
the exponential growth rate of Investor's capital process
when Market chooses a path with variation exponent not equal to two.
Finally in Section \ref{sec:discussions} we discuss
some aspects of our results and their implications for related fields.

\section{Preliminaries on coin-tossing games and Bayesian strategies}
\label{sec:preliminary}

In this section we summarize preliminary results  on coin-tossing games
and capital process of Bayesian strategies (\cite{shafer/vovk:2001},
\cite{ktt1}).  We also discuss  a one-to-one correspondence between
the set of probability distributions on the set of Reality's paths and the set
of Skeptic's strategies.  Finally we note the convexity of the
Kullback divergence with respect to its first argument.

In this paper we consider the coin-tossing game in the following form. In
the protocol the success probability $0<\rnp<1$ is given.   

\medskip
\noindent
\textsc{Coin-Tossing Game} \\
\textbf{Protocol:}

\parshape=6
\IndentI   \WidthI
\IndentI   \WidthI
\IndentII  \WidthII
\IndentII  \WidthII
\IndentII  \WidthII
\IndentI   \WidthI
\noindent
${\cK}_0 :=1$.\\
FOR  $n=1, 2, \dots$:\\
  Skeptic announces $M_n\in{\mathbb R}$.\\
  Reality announces $x_n\in \{0,1\}$.\\
  ${\cK}_n = {\cK}_{n-1} + M_n (x_n-\rnp)$.\\
END FOR

\medskip
If we write $M_n=M_n^1 - M_n^0$, then
$M_n(x_n-\rnp)$ can be rewritten as
\begin{equation}
\label{eq:redundant}
M_n(x_n-\rnp)= M_n^1 (x_n-\rnp) + M_n^0 \big((1-x_n) - (1-\rnp)\big).
\end{equation}
In this case we say that Skeptic bets $M_n^1$ on $x_n=1$  and
$M_n^0$ on $x_n=0$. Although  (\ref{eq:redundant}) is a redundant
expression, generalizations to multistep protocols in the next
section can be more transparently  understood in this form.
A path $\xi=x_1 x_2 \dots$ is an infinite sequence of Reality's moves
and the sample space $\Xi = \{\xi\}=\{0,1\}^\infty$ is the set of
paths.  $\xi^n = x_1 \dots x_n$ denotes the partial path of Reality's
moves up to round $n$.
Throughout this paper we use the notation
\[
s_n=x_1 + \dots + x_n, \qquad
\bar x_n =  \frac{s_n}{n}.
\]

Skeptic's strategy $\cal P$ is a set of functions
$
{\cal P} :  \xi^{n-1}\mapsto M_n 
$
which determines Skeptic's move $M_n$ at round $n$ based on Reality's moves
up to the previous round $\xi^{n-1}=x_1 \dots x_{n-1}$.   Given a
strategy $\cP$ 
\[
\cK_n^{\cP}(\xi)= \cK_0 + \sum_{i=1}^n M_i(\xi^{i-1}) (x_i - \rnp)
\]
denotes Skeptic's capital process when he uses $\cP$, 
starting with the initial capital of $\cK_0=1$.
Following the terminology in \cite{vovk-cont-1} we call  $\cP$ {\em prudent}
if $\cK_n^\cP(\xi) \ge 0$ for all $\xi$ and $n$, i.e., Skeptic's
capital is never negative irrespective of the moves of Reality.
We also say that Skeptic observes his collateral duty if he uses
a prudent strategy.
In this paper we require  that Skeptic's strategies are prudent.

We consider probability distributions on the set of paths
$\Xi=\{0,1\}^\infty$.  
In our framework
a probability distribution $Q$ on $\Xi$
is just a collection of consistent discrete probability distributions
$Q= \{Q_n ; n \ge 1\}$, where $Q_n$ is a discrete probability
distribution on $\Xi_n = \{0,1\}^n$  satisfying the  consistency condition
\begin{equation}
\label{eq:consistency}
Q_n(\xi^n)=Q_{n+1}(\xi^n 0) + Q_{n+1}(\xi^n 1), \qquad \forall n,
\forall \xi^n.
\end{equation}
Note that 
we do not need measure-theoretic extension of $Q$ to a
probability measure on a $\sigma$-field of $\Xi$.
Under the distribution $Q$, the conditional probability of $x_n=1$ 
given $\xi^{n-1} =x_1 \dots x_{n-1}$ with $Q_{n-1} (\xi^{n-1})>0$ is written as
\begin{equation}
\label{eq:conditional-probability}
p_n^Q = p_n^Q(\xi^{n-1})= 
\frac{Q_n (\xi^{n-1}1)}{Q_{n-1} (\xi^{n-1})}.
\end{equation}
If $Q_{n-1} (\xi^{n-1})=0$, then $p_n^Q$ is not defined.
We call the probability distribution of i.i.d.\ Bernoulli trials with success
probability $\rnp$
the {\em risk neutral measure} of the coin-tossing game.

Given a probability distribution $Q$, define a strategy $\cP=\cP_Q$
by
\begin{equation}
\label{eq:QtoP}
\cP_Q : \xi^{n-1}\mapsto M_n = \cK_{n-1} \frac{p_n^Q-
  \rnp}{\rnp(1-\rnp)}.
\end{equation}
A motivation of this definition is given in Appendix.
If we write $M_n=M_n^1-M_n^0$ as in (\ref{eq:redundant}), then 
$M_n^1$, $M_n^0$ are given as
\begin{equation}
\label{eq:redundant-M}
M_n^1 = \cK_{n-1}\frac{p_n^Q}{\rnp},  \quad
M_n^0 = \cK_{n-1}\frac{1-p_n^Q}{1-\rnp}.
\end{equation}
The capital process of $\cP_Q$ is explicitly written as follows
(\cite[Theorem 4.1]{ktt1}).
\begin{equation}
\label{eq:capital-Q}
\cK_n^{\cP_Q}(\xi^n) = \frac{Q(\xi^n)}{ \rnp^{s_n} (1-\rnp)^{n-s_n}}. 
\end{equation}
This is the likelihood ratio of $Q$ to the risk neutral measure at the 
realized  path $\xi^n$ up to round $n$.
In Appendix, we establish a one-to-one correspondence between
the set of probability distributions on 
$\Xi$ and the set of prudent strategies.
Therefore capital process of a prudent strategy
can always be expressed as (\ref{eq:capital-Q}).

In particular if we employ  a beta-binomial distribution 
\begin{equation}
\label{eq:beta-binomial}
Q(\xi^n)=\frac{\Gamma(a+b) \Gamma(a+s_n)\Gamma(b+n-s_n)}
   {\Gamma(a+b+n)\Gamma(a)\Gamma(b)} 
 \qquad \text{with}\ \ p_n^Q=\frac{a+s_{n-1}}{a+b+n-1}, 
\end{equation}
where $a,b>0$ are hyperparameters of the prior distribution, then by a
simple application of Stirling's formula, the resulting capital
process  denoted by $\cK_n^0$ behaves as
\begin{equation}
\label{eq:bayes1}
\log \cK_n^0= n D( \bar x_n \| \rnp) - O(\log n),
\end{equation}
where 
\[
D(p\Vert q)=p\log \frac{p}{q}+ (1-p)\log\frac{1-p}{1-q} 
\]
is the  Kullback divergence between two scalar probabilities $p$ and $q$.
Hence if $\bar x_n$ deviates from $\rnp$, then 
$D( \bar x_n \| \rnp)$ gives the average exponential growth rate of 
the capital process.  We call $D( \bar x_n \| \rnp)$ the main growth
rate (or simply the growth rate) of the log capital.

In the subsequent sections we
often use the convexity of Kullback divergence for probability vectors
with respect to its
first argument.
Let ${\bm p}=\{p_j\}_{j=1}^k$ and ${\bm
  q}=\{q_j\}_{j=1}^k$ be probability vectors and let
\[
D({\bm p} \| {\bm q}) = \sum_{j=1}^k p_j \log \frac{p_j}{q_j}
\]
denote the Kullback divergence between ${\bm p}$ and ${\bm q}$.
Let ${\bm p}_i=\{p_{ij}\}_{j=1}^k$, $i=1,2$, be probability
vectors and for $0<\lambda<1$ let $\bar {\bm p}=\lambda {\bm p}_1 +
(1-\lambda) {\bm p}_2=
\{\bar p_j\}_{j=1}^k$.
Then the following relation is easily obtained
\begin{equation}
 \lambda D({\bm p}_1 \Vert {\bm q}) + (1-\lambda) D({\bm p}_2 \Vert {\bm q})
-  
D(\bar{\bm p}\Vert {\bm q})
= \lambda D({\bm p}_1 \Vert \bar{\bm p}) + 
  (1-\lambda) D({\bm p}_2 \Vert \bar{\bm p})  \ge 0 .
\label{eq:KL-ineq}
\end{equation}
The left-hand side can also be written as
\[
\lambda D({\bm p}_1 \Vert {\bm q}) + (1-\lambda) D({\bm p}_2 \Vert {\bm q})
- D(\bar{\bm p}\Vert {\bm q})
=\sum_{j=1}^k \bar p_j  D\Big(\frac{\lambda p_{1j}}{\bar p_j} \Big\Vert \lambda\Big) \ge 0. 
\]

\section{Priors for higher order patterns and multistep strategies}
\label{sec:patterns}
As discussed in Section \ref{sec:intro}
there may be some deviating patterns 
from independent Bernoulli trials 
in the Reality's moves  $x_1 x_2 \dots$, 
although the path $\xi=x_1 x_2 \dots$  satisfies SLLN
$\lim_{n \to \infty} \bar x_n = \rnp$.
The strategy considered in \cite{ktt1} is based only on
$s_n$ and it can not exploit these patterns.   Skeptic can increase
his capital by strategies exploiting patterns not reflected in $\bar x_n$.
In the following we investigate two types of such  
non-randomness or higher order patterns. The first is the block type 
pattern  and the second is the Markovian type pattern.  We give
multistep Bayesian strategies which effectively exploit these patterns.

\subsection{Block patterns}
\label{subsec:block}
For clarity of presentation, we first consider sequence of pairs
(i.e.\ blocks of length 2) and later generalize the results to blocks of
arbitrary length.

Consider the sequence of pairs $(x_1  x_2)(x_3  x_4)\ldots
(x_{2n-1}  x_{2n})$ of Reality's moves and denote the number 
of the pairs $(11), (1 0), (0 1),  (0 0)$ among the
first $n$ blocks by
$m_n^{11}$, $m_n^{10}$, $m_n^{01}$, $m_n^{00}$, respectively. If the
sequence is random, i.e., the Reality's moves are i.i.d.\
Bernoulli trials with success probability $\rnp$, then
we will have
\begin{align*}
\lim_{n\to \infty} \frac{m_n^{ij}}{n} = \rnp^{ij}, \quad i, j = 0, 1,
\end{align*}
where $\rnp^{11} = (\rnp)^2,\ 
\rnp^{10} = \rnp^{01} = \rnp(1 - \rnp),\ \rnp^{00} = (1 - \rnp)^2$. 
%
We construct a strategy for which $\limsup_n \cK_n = \infty$ whenever 
$\limsup_n |m_n^{ij}/n - \rnp^{ij}| > 0$ at least for one $(i, j)$.

For this purpose,  at the $(2n - 1)$-th round $(n = 1, 2, \dots)$  
Skeptic chooses
four amounts  $M_n^{11}, M_n^{10}, M_n^{01}, M_n^{00}$ and bet them on 
$(x_{2n-1}  x_{2n}) = (1  1),  (1 0), (0 1), (0 0)$, respectively. 
Then we have 
\begin{equation*}
\cK_{2n} = \cK_{2n-2} + \sum_{i,j\in\{0,1\}} M_n^{ij}(z_n^{ij} - \rnp^{ij}),
\end{equation*}
where
\[
z_n^{11} = x_{2n-1} \times  x_{2n} = 
\begin{cases} 1 , &  \textrm{if}\ (x_{2n-1} x_{2n}) = (1 1) \\
              0, & \textrm{otherwise}\ 
\end{cases}
\]
and other $z_n^{ij}$, 
$i, j = 0, 1$, are defined similarly. 
%
Thus we define a derived capital process $\cK_n^*=\cK_{2n}$
with the protocol
\[
\cK_n^* = \cK_{n-1}^* + \sum_{i,j\in \{0,1\}} M_n^{ij}(z_n^{ij} - \rnp^{ij}),
\quad n = 1, 2, \dots,  \qquad
(\cK_0^* = 1).
\]

As a natural generalization of the beta-binomial distribution 
treated in \cite{ktt1}, let us take the  
Dirichlet-multinomial distribution as $Q(\xi^{2n})=Q(z_1\dots z_n)$. 
Then the corresponding strategy 
is given by (cf.\ (\ref{eq:redundant-M}) and (\ref{eq:beta-binomial}))
\[M_n^{ij}= \cK_{n-1}^*
\frac{m_{n-1}^{ij} + c^{ij}}{\rnp^{ij}(n - 1 + c)},\quad i, j = 0, 1,
\]
where $c^{ij}$'s are positive hyperparameters of the Dirichlet prior
and $\sum_{i,j\in\{0.1\}} c^{ij} = c$.
The capital process $\cK_n^*
=\cK_{2n}^{\cP_Q}$
for this strategy is given by 
\begin{align*}
\cK_n^* &= \frac{Q(\xi^{2n})}{\prod_{i,j} (\rnp^{ij})^{m_n^{ij}}} 
= \frac{\prod_{i,j} \Gamma(m_n^{ij} + c^{ij})/\Gamma(c^{ij})}
{\big(\Gamma(n + c)/\Gamma(c)\big)\prod_{i,j} (\rnp^{ij})^{m_n^{ij}}}\\ 
&= \frac{\Gamma(c)\prod_{i,j} \Gamma(m_n^{ij} + c^{ij})}
{\prod_{i,j} \Gamma(c^{ij})\Gamma(n + c)
\prod_{i,j} (\rnp^{ij})^{m_n^{ij}}},
\end{align*}
where in the products $i,j$ range over $\{0,1\}$.

We evaluate the asymptotic behavior of this capital. 
Denote $m_n^{ij}/n = \hat p_n^{ij}$,\ $i, j = 0, 1$.
Then 
as in (\ref{eq:bayes1}) we have
\[
\log \cK_n^* 
= n\sum_{i,j\in\{0,1\}} \hat p_n^{ij} \log \frac{\hat p_n^{ij}}{\rnp^{ij}} - O(\log n).
\]
Hence for even $n$ the original capital process $\cK_n^{\cP_Q}=\cK_{n/2}^*$ is
written as
\[
\log \cK_n^{\cP_Q} = \log \cK_{n/2}^* 
= \frac{n}{2}D\big(\{\hat p_n^{ij}\} \bigm\| \{\rnp^{ij}\}\big) - O(\log n).
\]

Now if we neglect the pairwise block patterns and apply the strategy of
\cite{ktt1} based on $s_n$ only,  
then corresponding capital process  $\cK_n^0$ 
behaves as (\ref{eq:bayes1}).
We compare $\cK_n^{\cP_Q}$ and $\cK_n^0$.
By 
(\ref{eq:KL-ineq}) with $\lambda=1/2$ we have
\begin{equation}
\label{eq:block2-comparison}
\log \cK_n^{\cP_Q} - \log \cK_n^0
= \frac{n}{2}D\big((\hat p_n^{11}, \hat p_n^{10}, \hat p_n^{01}, \hat p_n^{00})
\bigm\| 
(\hat \rnp_n^{11}, \hat \rnp_n^{10}, \hat \rnp_n^{01}, \hat \rnp_n^{00})
\big)  - O(\log n),
\end{equation}
where $\hat\rnp_n^{11} = (\bar x_n)^2, 
\hat\rnp_n^{10} = \hat\rnp_n^{01} = \bar x_n(1 - \bar x_n), 
\hat\rnp_n^{00} = (1 - \bar x_n)^2$. 
We see that $\cK_n^{\cP_Q}$ exploits the pairwise block type 
non-randomness more effectively than $\cK_n^0$ by the amount half the 
Kullback divergence given in the right-hand side of 
(\ref{eq:block2-comparison}).

So far we have only considered even $n$.  This is sufficient
for analyzing the asymptotic behavior of $\cK_n^{\cP_Q}$.
For completeness we discuss $\cK_n^{\cP_Q}$ for odd $n$.
We can decompose Skeptic's bet $M_n^{ij}$
on  $z_n^{ij}=1$ to  
the $(2n-1)$-th round and the $2n$-th round as follows.
\begin{itemize}
\setlength{\itemsep}{0pt}
\item[1)]
 at the $(2n-1)$-th round, Skeptic bets 
$\rnp M_n^{11} + (1-\rnp) M_n^{10}$ on $x_{2n-1}=1$  and 
$\rnp M_n^{01} + (1-\rnp) M_n^{00}$ on $x_{2n-1}=0$, 
\item[2a)] if $x_{2n-1}=1$, then at the $2n$-th round he bets  
$M_n^{11}$ on $x_{2n}=1$ and
$M_n^{10}$ on $x_{2n}=0$,
\item[2b)] if $x_{2n-1}=0$, then at the $2n$-th round he bets 
$M_n^{01}$ on $x_{2n}=1$ and
$M_n^{00}$ on $x_{2n}=0$.
\end{itemize}
Denote $m_n^{i+}=\sum_{j=0}^1 m_n^{ij}$ and 
$c^{i+}=\sum_{j=0}^1 c^{ij}$, $i=0,1$.
Then  the capital at an odd round $\cK_{2n+1}^{\cP_Q}$ is written as follows.
\begin{align*}
\cK_{2n+1}^{\cP_Q} &= \cK_{2n}^{\cP_Q}\times
\begin{cases} (m_n^{1+} +c^{1+})/(\rnp (n-1+c)) , &  \textrm{if}\ x_{2n+1}=1 \\
              (m_n^{0+} +c^{0+})/((1-\rnp) (n-1+c)), & \textrm{if}\ x_{2n+1}=0
\end{cases} \\
&=\cK_{2n}^{\cP_Q}\times \frac{\Gamma(c)\Gamma(m_n^{1+}+
  c^{1+}+x_{2n+1}) \Gamma(m_n^{0+} +   c^{0+} +1-x_{2n+1})}
 {\rnp^{x_{2n+1}}(1-\rnp)^{1-x_{2n+1}} \Gamma(n+c)
\Gamma(m_n^{1+}+   c^{1+}) \Gamma(m_n^{0+} + c^{0+})
}.
\end{align*}

\medskip
We can construct a similar strategy for the sequence of pairs 
$(x_2 x_3)\dots (x_{2n} x_{2n+1})$  and can also combine these two 
strategies by splitting the initial capital into two equal parts and
applying the corresponding strategy for each of them. 
Let $\cK^{B}=\cK^{\cP_Q}$ denote the capital process considered so
far and let $\tilde \cK^{B}$ denote the similar capital process based on
$(x_2 x_3)(x_4 x_5)\dots$.  Then the capital process of the combined
strategy is written as $\cK=(1/2)(\cK^{B}+\tilde\cK^{B})$.
Hence 
$\limsup_n \cK_n = \infty$, if the 
relative frequency of consecutive pairs
$(x_{2n-1} x_{2n})$,  $(x_{2n} x_{2n+1})$,
taking the value $(ij)$
does not converge to 
$\{\rnp^{ij}/2\}$ for some $(ij)$.

\medskip
We have discussed blocks of length two for notational simplicity.
The above derivation can be extended to $n$ blocks of consecutive 
$k$-tuples, and hereafter we outline the procedure. 
Let $\bm{x}_1^k, \bm{x}_2^k, \dots, \bm{x}_n^k$ be the first $n$ blocks 
of $k$-tuple with
\[
\bm{x}_m^k = (x_{k(m-1)+1}, x_{k(m-1)+2}, \dots, x_{km}),\ m = 1, \dots, n,
\]  
and let $m_n^{\bm{\epsilon}_k}$ , 
$\bm{\epsilon}_k = \epsilon_1 \dots \epsilon_k$,  $\epsilon_i = 1$
or $0$, denote 
the number of the consecutive $k$-tuple $\bm{\epsilon}_k$
among the first $n$ blocks. 
Also denote
\[
\rnp^{\bm{\epsilon}_k} 
= \rnp^{\sum_{i=1}^k \epsilon_i}(1 - \rnp)^{k-\sum_{i=1}^k \epsilon_i} .
\] 
If the sequence is random, 
we will have
$
\lim_{n\to \infty} m_n^{\bm{\epsilon}_k}/n = \rnp^{\bm{\epsilon}_k} 
$
for all $\bm{\epsilon}_k\in \{0,1\}^k$.

At the $k(n - 1) + 1$-st round $(n = 1, 2,
\dots)$, 
Skeptic chooses $2^k$ amounts $M_n^{\bm{\epsilon}_k}$ and bet them on 
$\bm{x}_n^k=\bm{\epsilon}_k$. Then we have 
\[
\cK_{kn} = \cK_{k(n-1)} + \sum_{\bm{\epsilon}_k \in \{0,1\}^k} 
M_n^{\bm{\epsilon}_k}
(z_n^{\bm{\epsilon}_k} - \rnp^{\bm{\epsilon}_k}),
\]
where
\[
z_n^{\bm{\epsilon}_k} = 
\begin{cases} 1 , &  \textrm{if}\ \bm{x}_n^k = \bm{\epsilon}_k \\
              0, & \textrm{otherwise}.\ 
\end{cases}
\]
Thus the derived game 
$\cK_n^* =\cK_{kn}$ is defined with the protocol
\[
\cK_n^* = \cK_{n-1}^* + \sum_{\bm{\epsilon}_k\in \{0,1\}^k} 
M_n^{\bm{\epsilon}_k}(z_n^{\bm{\epsilon}_k} - \rnp^{\bm{\epsilon}_k}),
\ n = 1, 2, \dots,   \qquad (\cK_0^* = 1), 
\]
where $\sum_{\bm{\epsilon}_k\in \{0,1\}^k} z_n^{\bm{\epsilon}_k} = 1.$ 

If we take the Dirichlet-multinomial distribution as
$Q(\xi^{kn})=Q(z_1^{\bm{\epsilon}_k}\dots z_n^{\bm{\epsilon}_k})$, then the corresponding
strategy 
is given by
\[
M_n^{\bm{\epsilon}_k}=
{\cK_{n-1}^*} 
\frac{m_{n-1}^{\bm{\epsilon}_k} + c^{\bm{\epsilon}_k}}
{\rnp^{\bm{\epsilon}_k}(n - 1 + c)},
\]
where $\forall c^{\bm{\epsilon}_k} > 0,\ 
\sum_{\bm{\epsilon}_k\in \{0,1\}^k} c^{\bm{\epsilon}_k} = c$, are the
hyperparameters of the Dirichlet prior.
The capital process for this strategy is  
\[
\cK_{kn}^{\cP_Q}=
\cK_n^* = \frac{Q(\xi^{kn})}{\prod_{\bm{\epsilon}_k} 
(\rnp^{\bm{\epsilon}_k})^{m_n^{\bm{\epsilon}_k}}} 
= \frac{\Gamma(c)\prod_{\bm{\epsilon}_k} 
\Gamma(m_n^{\bm{\epsilon}_k} + c^{\bm{\epsilon}_k})}
{\prod_{\bm{\epsilon}_k} \Gamma(c^{\bm{\epsilon}_k})\Gamma(n + c)
\prod_{\bm{\epsilon}_k} (\rnp^{\bm{\epsilon}_k})^{m_n^{\bm{\epsilon}_k}}}.
\]
We evaluate the asymptotic behavior of this capital. 
Denote  $m_n^{\bm{\epsilon}_k}/n =\hat p_n^{\bm{\epsilon}_k}$. 
Then 
\begin{equation}
\label{eq:block-general}
\log \cK_n^* = n D\big(\{\hat p_n^{\bm{\epsilon}_k} \}  
\bigm\| \{\rnp^{\bm{\epsilon}_k}\}\big) - O(\log n)
= n\sum_{\bm{\epsilon}_k\in \{0,1\}^k} \hat p_n^{\bm{\epsilon}_k} 
\log \frac{\hat p_n^{\bm{\epsilon}_k}}{\rnp^{\bm{\epsilon}_k}} - O(\log n).
\end{equation}
Hence in the context of the original game, for $n$ which is a multiple
of $k$, we have
\begin{equation*}
\log \cK_n^{\cP_Q} = \log \cK_{n/k}^* 
= \frac{n}{k} D\big(\{\hat p_n^{\bm{\epsilon}_k} \} \bigm\| 
\{\rnp^{\bm{\epsilon}_k} \} \big) - O(\log n).
\end{equation*}

Similar strategies can be constructed for the $n$ blocks of the $k$-tuples 
with shift $a$ 
\begin{equation*}
\bm{x}_{m+a}^k = (x_{k(m-1)+1+a}, x_{k(m-1)+2+a}, \dots, x_{km+a}),\quad
1\le m\le n,\ a = 1,\dots, k-1.
\end{equation*}
We next combine $k$ kinds of these strategies by splitting the initial
capital into $k$ equal parts and applying the corresponding strategy
for each of them. 
Let $\cK_n$ denote the resulting capital process.
Then $\limsup_n \cK_n = \infty$, 
if for any  $\bm{\epsilon}_k$
the relative frequency of any of 
$\bm{x}_{m}^k, \bm{x}_{m+1}^k,\dots,$ $\bm{x}_{m+k-1}^k$ \ 
($1\le m\le n$) taking  the value 
$\bm{\epsilon}_k$ does not converge to $\rnp^{\bm{\epsilon}_k}/k$.

It is of interest to compare the growth rates of $\cK_n^{\cP_Q}$ for
different block lengths $k$.  
Note that block pattern strategy of length $k$ also depends on the
shift $a$.  If the empirical distribution of $\{\bm{\epsilon}_k\}$ 
is different for $a$, then the main growth rate of a block
strategy depends also on $a$.  For simplicity we mainly consider the case
that the empirical distribution $F_{n,a}$ 
of 
$\{\bm{x}_{m+a}^k\}_{m=1}^n$ 
is the same 
for different shifts $a$,  in the sense that
the total variation distance between 
$F_{n,a}$ and $F_{n,a'}$  converges to 0 for
all $a\neq a'$.
We call this case ``homogeneous with respect to the shifts''.
We show in Section \ref{sec:comparison-of-rate}
that the main growth rate of (\ref{eq:block-general})
is non-decreasing in $k$ under the assumption 
of homogeneity with respect to the shifts.

\subsection{Markovian patterns}
\label{subsec:Markov}
We can construct another strategy which exploits non-randomness 
in the moves of Reality.  
Such a  procedure can be given  as
\[
M_1 = 0,\quad
M_n = 
\begin{cases} M_n^+ , &  \textrm{if}\ x_{n-1} = 1 \\
              M_n^-, & \textrm{if}\ x_{n-1} = 0
\end{cases}
\quad n = 2, 3, \dots,
\]
where $M_n^+$ and $M_n^-$ can have different values.  This is a
first-order Markovian strategy, which incorporates the information on 
the last move $x_{n-1}$ of Reality.  As in the previous subsection,
for clarity of presentation we first consider the first-order Markovian
strategy and then extend it to higher-order Markovian strategies.

Let $q_n^1=s_n$ and $q_n^0=n-s_n$.
We also denote the numbers of pairs $(x_{i-1} x_i) = (1
1),(10), (01), (00),\ i = 2, \dots, n$ by $q_n^{11},\  q_n^{10},\
q_n^{01},\ q_n^{00}$, respectively.
(For $n=1$,  let $0=q_1^{11}=q_1^{10}=q_1^{01}=q_1^{00}$.) 
We take the beta-binomial distribution with parameters $a,b>0$
for  $Q(x_i \mid x_{i-1} = 1)$  and
$Q(x_i \mid  x_{i-1} = 0)$, $i\ge 2$.  
(The initial distribution is taken as $Q(1)=\rnp=1-Q(0)$.)
Then the corresponding
strategy for $i\ge 2$ is
given by
\begin{align*}
M_i^+ = \cK_{i-1} \frac{p_i^+ - \rnp}{\rnp(1 - \rnp)},\quad 
M_i^- = \cK_{i-1} \frac{p_i^- - \rnp}{\rnp(1 - \rnp)},
\end{align*}
with
\begin{align*}
p_i^+ = \frac{q_{i-1}^{11} + a}{q_{i-1}^{11} + q_{i-1}^{10} + a + b},\quad 
p_i^- = \frac{q_{i-1}^{01} + a}{q_{i-1}^{01} + q_{i-1}^{00} + a + b}.
\end{align*}
For $i=1$, we let $M_1^+=M_1^-=0$.
The capital $\cK_n=\cK_n^{\cP_Q}$  for this strategy is given by
\begin{align*}
\cK_n &= \frac{\Gamma(a + b)\Gamma(q_n^{11} + a)
\Gamma(q_n^{10} + b)}
{\Gamma(a)\Gamma(b)\Gamma(q_n^{11} + q_n^{10} + a + b)
\rnp^{q_n^{11}}(1 - \rnp)^{q_n^{10}}}\times
\frac{\Gamma(a + b)\Gamma(q_n^{01} + a)
\Gamma(q_n^{00} + b)}
{\Gamma(a)\Gamma(b)\Gamma(q_n^{01} + q_n^{00} + a + b)
\rnp^{q_n^{01}}(1 - \rnp)^{q_n^{00}}}\\
&= \frac{\Gamma(a + b)^2\Gamma(q_n^{11} + a)
\Gamma(q_n^{10} + b)\Gamma(q_n^{01} + a)\Gamma(q_n^{00} + b)}
{\Gamma(a)^2\Gamma(b)^2\Gamma(q_n^{11} + q_n^{10} + a + b)
\Gamma(q_n^{01} + q_n^{00} + a + b)
\rnp^{q_n^{11}+q_n^{01}}(1 - \rnp)^{q_n^{10}+q_n^{00}}}.
\end{align*}

We evaluate the asymptotic behavior of this capital process. 
Write $q_n^1/n = \hat p_n = \bar x_n$, $q_n^{11}/q_n^1 =r_n^1$, and
$q_n^{01}/q_n^0 =r_n^0$.
Then as in (\ref{eq:bayes1})
we have
\begin{equation}
\label{eq:markov2}
\log \cK_n = n \hat p_n D(r_n^1\| \rnp) + 
n (1 - \hat p_n)D(r_n^0 \| \rnp )
- O(\log n).
\end{equation}
Hence if either $\limsup_n |r_n^1- \rnp| > 0$ or 
$\limsup_n |r_n^0- \rnp| > 0$ 
then 
$\limsup_n \cK_n =\infty$.

Now we compare (\ref{eq:markov2}) to the capital process 
(\ref{eq:bayes1}) of the strategy based on $s_n$ only.
By counting the number of pairs we have
$q_n^{11} + q_n^{01} = q_n^1 - x_1$ and
\[
\hat p_n r_n^1 + 
(1-\hat p_n) r_n^0= \frac{q_n^{11}}{n} + \frac{q_n^{01}}{n}=\hat p_n 
  - \frac{x_1}{n} = \hat p_n - O(1/n).
\]
Then by (\ref{eq:KL-ineq}) with $\lambda=\hat p_n$ we have
\begin{align}
\log \cK_n - \log \cK_n^0
&=  n  \Big[
\hat p_n D(r_n^1 \| \rnp) +  
(1-\hat p_n) D(r_n^0 \| \rnp)   
 -  D(\hat p_n \| \rnp)\Big] - O(\log n)  \nonumber \\
&=  n 
\Big[
\hat p_n D(r_n^1 \| \hat p_n) +  
(1 - \hat p_n)D(r_n^0 \| \hat p_n )\Big]
- O(\log n).
\label{eq:markov2-comparison}
\end{align}
We again see that $\log \cK_n$ exploits the first-order Markovian  
non-randomness more effectively than $\log \cK_n^0$ by the amount given above.

The above first-order procedure can be extended to the $k$-th order
procedure based on sequence $\tilde{\bm{x}}_{n-1}^k = x_{n-k} \dots x_{n-1}$ 
of length $k$ preceding $x_n$.
We hereafter outline 
the procedure. 

Let $q_n^{\bm{\epsilon}_k}$ denote 
the number of consecutive $k$-tuple $\bm{\epsilon}_k$ 
in $\xi^n=x_1\dots x_n$. 
For $(k + 1)$-tuples
$\bm{\epsilon}_k 1 = \epsilon_1\dots \epsilon_k 1$ 
and 
$\bm{\epsilon}_k 0 = 
\epsilon_1,\dots, \epsilon_k 0$
we similarly define
$q_n^{\bm{\epsilon}_k1}$ and 
$q_n^{\bm{\epsilon}_k0}$.
We take the beta-binomial distribution for 
$Q(x_i | \tilde{\bm{x}}_{i-1}^k = \bm{\epsilon}_k)$, $i\ge k+1$. 
(The initial distribution $Q(\xi^k)$ up to round $k$ 
is taken as the risk neutral measure.)
The corresponding strategy is
\begin{align*}
M_i = 0,\quad 1 \le i \le k,\quad
M_i = M_i^{\bm{\epsilon}_k} = 
\cK_{i-1}
\frac{p_i^{\bm{\epsilon}_k} - \rnp}{\rnp(1 - \rnp)},\quad 
\textrm{if}\ \tilde{\bm{x}}_{i-1}^k = \bm{\epsilon}_k,
\quad i \ge k + 1,
\end{align*}
with
\begin{align*}
p_i^{\bm{\epsilon}_k} = 
\frac{q_{i-1}^{\bm{\epsilon}_k1} + a}{q_{i-1}^{\bm{\epsilon}_k1} + q_{i-1}^{\bm{\epsilon}_k0} + a + b},
\quad a, b > 0.
\end{align*}
The capital process $\cK_n=\cK_n^{\cP_Q}$ for this strategy is written as 
\begin{align*}
\cK_n &= \prod_{\bm{\epsilon}_k \in \{0,1\}^k}
\frac{\Gamma(a + b)\Gamma(q_n^{\bm{\epsilon}_k1} + a)
\Gamma(q_n^{\bm{\epsilon}_k0} + b)}
{\Gamma(a)\Gamma(b)\Gamma(q_n^{\bm{\epsilon}_k} + a + b)
\rnp^{q_n^{\bm{\epsilon}_k1}}(1 - \rnp)^{q_n^{\bm{\epsilon}_k0}}}\\
&= \frac{\Gamma(a + b)^{2^k}}{\Gamma(a)^{2^k}\Gamma(b)^{2^k}}
\prod_{\bm{\epsilon}_k\in \{0,1\}^k}
\frac{\Gamma(q_n^{\bm{\epsilon}_k1} + a)
\Gamma(q_n^{\bm{\epsilon}_k0} + b)}
{\Gamma(q_n^{\bm{\epsilon}_k} + a + b)
\rnp^{q_n^{\bm{\epsilon}_k1}}(1 - \rnp)^{q_n^{\bm{\epsilon}_k0}}}.
\end{align*}
We evaluate the asymptotic behavior of this capital. 
Denote $\hat p_n^{M,\bm{\epsilon}_k}=q_n^{\bm{\epsilon}_k}/n$, 
$r_n^{\bm{\epsilon}_k}= q_n^{\bm{\epsilon}_k1}/q_n^{\bm{\epsilon}_k}$.
Note that $\hat p_n^{M,\bm{\epsilon}_k}$ differs from 
$\hat p_n^{\bm{\epsilon}_k}$ in 
(\ref{eq:block-general}), because the latter only looks at the 
relative frequency of $\bm{\epsilon}_k$ among non-overlapping blocks of length $k$.
As $n \to \infty$, we have
\begin{equation}
\label{eq:Markov-general}
\log \cK_n %
= n \sum_{\bm{\epsilon}_k\in \{0,1\}^k} \hat p_n^{M,\bm{\epsilon}_k}
D\big( r_n^{\bm{\epsilon}_k}\|  \rnp\big) - O(\log n).
\end{equation}
Hence if $\limsup_n | r^{\bm{\epsilon}_k}-\rnp| > 0$ for  any $\bm{\epsilon}_k$,  
then 
$\limsup_n \cK_n = \infty$.
In the next section we show that the first term on 
the right-hand side of
(\ref{eq:Markov-general}) is non-decreasing in $k$.

\subsection{Relations between block strategies and Markovian strategies}
\label{sec:comparison-of-rate}

In (\ref{eq:block2-comparison}) and (\ref{eq:markov2-comparison}) we
saw that the block strategy of length two and the first-order
Markovian strategy have better main growth rates than the strategy
based $s_n$ only. In this section we give results on the comparison of
main growth rates for general $k$, which is the block size for block
strategies and the order for Markovian strategies.  

For Markovian strategies we show that a larger $k$ gives a better growth
rate.  Concerning block strategies we show that the same result holds  under 
the assumption of homogeneity with respect to shifts and furthermore
that Markovian strategy of order $k-1$ gives a better growth rate than
the block strategy of length $k$.

We first consider Markovian strategies.
Let $\cK_{n}^{M,k}$ denote the capital process of $k$-th order
Markovian strategy.
In (\ref{eq:Markov-general}), for a given
$\bm{\epsilon}_{k-1}=\epsilon_2 \dots  \epsilon_k$ consider
the sum of two terms involving 
$1\bm{\epsilon}_{k-1}=1 \epsilon_2 \dots \epsilon_k$
and 
$0\bm{\epsilon}_{k-1}=0 \epsilon_2 \dots \epsilon_k$.
We have
\[
\hat p_n^{M,1\bm{\epsilon}_{k-1}} + \hat p_n^{M,0\bm{\epsilon}_{k-1}}
=\hat p_n^{M,\bm{\epsilon}_{k-1}} + O(1/n),
\]
where $O(1/n)$ is due to the counting problem at the
end of the sequence $\xi^n$.  Also
\[
\hat p_n^{M,1\bm{\epsilon}_{k-1}} r_n^{1\bm{\epsilon}_{k-1}}
+ 
\hat p_n^{M,0\bm{\epsilon}_{k-1}} r_n^{0\bm{\epsilon}_{k-1}}
=
\hat p_n^{M,\bm{\epsilon}_{k-1}} r_n^{\bm{\epsilon}_{k-1}} +
O(1/n).
\]
Then by (\ref{eq:KL-ineq}) 
\begin{align}
&\hat p_n^{M,1\bm{\epsilon}_{k-1}}
D(r_n^{1\bm{\epsilon}_{k-1}} \| \rnp) + 
\hat p_n^{M,0\bm{\epsilon}_{k-1}}
D(r_n^{0\bm{\epsilon}_{k-1}} \| \rnp) 
- \hat p_n^{M,\bm{\epsilon}_{k-1}}
D(r_n^{1\bm{\epsilon}_{k-1}} \| \rnp)  \nonumber \\
& \qquad\qquad =
\hat p_n^{M,1\bm{\epsilon}_{k-1}}
D(r_n^{1\bm{\epsilon}_{k-1}} \| r_n^{\bm{\epsilon}_{k-1}})+
\hat p_n^{M,0\bm{\epsilon}_{k-1}} 
  D(r_n^{0\bm{\epsilon}_{k-1}} \| r_n^{\bm{\epsilon}_{k-1}}) +
  O(1/n). 
\end{align}
Summing up over ${\bm{\epsilon}_{k-1}}\in \{0,1\}^{k-1}$ we have
\begin{align}
& \log \cK_{n}^{M,k} - \log \cK_{n}^{M,k-1}  \nonumber\\
& \qquad
=
n \sum_{\bm{\epsilon}_{k-1}\in \{0,1\}^{k-1}} \big[\hat p_n^{M,1\bm{\epsilon}_{k-1}} 
  D(r_n^{1\bm{\epsilon}_{k-1}} \| r_n^{\bm{\epsilon}_{k-1}}) 
 + \hat p_n^{M,0\bm{\epsilon}_{k-1}} 
  D(r_n^{0\bm{\epsilon}_{k-1}} \| r_n^{\bm{\epsilon}_{k-1}}) 
\big]
- O(\log n).
\label{eq:Markov-better-k}
\end{align}
Therefore the growth rate of the  Markovian strategy of order $k$ is
larger than that of order $k-1$ by the amount
shown on the right-hand side.

Next we consider block strategies.  
For $j=0,1$, define
\[
\hat p_n^{j\mid \bm{\epsilon}_{k-1}}
=\frac{\hat p_n^{M,\bm{\epsilon}_{k-1}j}}{\hat p_n^{M,\bm{\epsilon}_{k-1}}},
 \qquad
\rnp^{j\mid {\epsilon}_{k-1}}=
\begin{cases} \rnp, & \textrm{if}\ j=1\\
             1-\rnp, & \textrm{if}\ j=0.
\end{cases}
\]
Then 
\[
\sum_{\bm{\epsilon}_k\in \{0,1\}^k} \hat p_n^{M,\bm{\epsilon}_k} 
\log \frac{\hat p_n^{M,\bm{\epsilon}_k}}{\rnp^{\bm{\epsilon}_k}}
=\sum_{\bm{\epsilon}_k\in \{0,1\}^{k-1}} \hat p_n^{M,\bm{\epsilon}_{k-1}}
  \sum_{j=0}^1 \hat p_n^{j\mid {\bm{\epsilon}_{k-1}}}
   \log \frac{ \hat p_n^{j\mid {\bm{\epsilon}_{k-1}}}}
    {\rnp^{j\mid {\bm{\epsilon}_{k-1}}}}
+ \sum_{\bm{\epsilon}_k\in \{0,1\}^{k-1}}  \hat p_n^{M,\bm{\epsilon}_{k-1}}
\log \frac{\hat p_n^{M,\bm{\epsilon}_{k-1}}}{\rnp^{\bm{\epsilon}_{k-1}}} .
\]
Now 
$\hat p_n^{1\mid \bm{\epsilon}_{k-1}}=r_n^{\bm{\epsilon}_{k-1}}$ and 
under the the assumption of homogeneity with respect to shifts, the relative frequency
$\hat p_n^{\bm{\epsilon}_k}$ of $\bm{\epsilon}_k$ 
is the same for different shifts $a$ and this implies that
$\hat p_n^{\bm{\epsilon}_{k-1}} = \hat p_n^{M,\bm{\epsilon}_{k-1}} + o(1)$.
Therefore
\[
D\big(\{\hat p_n^{\bm{\epsilon}_k} \}  
\bigm\| \{\rnp^{\bm{\epsilon}_k}\}\big)=
\sum_{\bm{\epsilon}_k\in \{0,1\}^k} \hat p_n^{M,\bm{\epsilon}_{k-1}}
D\big(r_n^{\bm{\epsilon}_{k-1}}\|  \rnp\big) + 
D\big(\{\hat p_n^{\bm{\epsilon}_{k-1}} \}  
\bigm\| \{\rnp^{\bm{\epsilon}_{k-1}}\}\big) + o(1).
\]
By induction on $k$ we obtain
\[
k  \log \cK_n^{B,k}
=   \log \cK_n^{M,k-1} +  (k-1)\log \cK_n^{B,k-1} + o(n)
=\sum_{i=0}^{k-1} \log \cK_n^{M,i}
+ o(n),
\]
or
\begin{equation}
\label{eq:block-markov-average}
\log \cK_n^{B,k} = \frac{1}{k}\sum_{i=0}^{k-1} \log \cK_n^{M,i}
+ o(n),
\end{equation}
where $\cK_n^{M,0}=\cK_n^0$ is the capital process of the strategy 
based on $s_n$ only.  
In (\ref{eq:Markov-better-k})
we saw that
the growth rate of Markovian strategy is non-decreasing in $k$.
It follows that under the assumption of homogeneity
with respect to shifts the growth rate of the block strategy is also non-decreasing in $k$
and furthermore the growth rate of the Markovian strategy of order
$k-1$ is better than that of the block strategy of length $k$.

We should note that if the homogeneity does not hold, then the growth
rate of the block strategy of length $k$ might be better than that of
the Markovian strategy of order $k-1$.  If we divide the initial
capital into $k$ equal parts corresponding to each shift, 
then the combined capital process is the
arithmetic average of the capital process for different shifts.  
The growth rate of the combined capital equals the maximum of 
capital processes for different shifts and the maximum might be better than that of
the Markovian strategy or order $k-1$.  The question of homogeneity comes up again in
consideration of the asset trading game in continuous time.

\subsection{Universal Bayesian Skeptic by mixture of priors}
\label{subsec:universal}
When we incorporate the strategies developed in the previous 
subsections, we get a strategy which can exploit any block or 
Markovian patterns of any length deviating from independent Bernoulli trials.

Let $\cP^{B,k}$ be the Bayesian strategy which exploits the $k$-th
order block patterns constructed in Section \ref{subsec:block}, and
let $\cP^{M,k}$ be the Bayesian strategy which exploits the $k$-th
order Markovian patterns constructed in Section \ref{subsec:Markov}.
At first we divide the initial capital 
$\cK_0 = 1$ into equal two parts $\cK_{0B} = 1/2$ and $\cK_{0M} = 1/2$, 
and further divide $\cK_{0B} = 1/2$ into countably infinite accounts 
with positive initial capitals $c_{B1}, c_{B2}, \dots,\ \sum_{k=1}^\infty c_{Bk} = 1/2$, and also divide 
$\cK_{0M} = 1/2$ similarly as
$c_{M1}, c_{M2}, \dots,\ \sum_{k=1}^\infty c_{Mk} = 1/2$. 
We apply the strategy $\cP^{B,k}$ to the $k$-th account $c_{Bk}$  
and apply the strategy $\cP^{M,k}$ to the $k$-th account $c_{Mk}$,
respectively. The resulting ``universal'' strategy 
\begin{equation*}
\cP^* = \cP^*_B + \cP^*_M = 
\sum_{k=1}^\infty c_{Bk}\cP^{B,k} + \sum_{k=1}^\infty c_{Mk}\cP^{M,k}
\end{equation*}
can exploit any block or Markovian pattern of any length. 

We shall formulate and state this fact within the framework of measure-theoretic
probability in order to clarify  the connection to
the universal source coding in information theory.  
For simplicity of statement we consider
the coin-tossing game with $\rnp=1/2$ and use the base two logarithm.
Let $\{p^{\bm{\epsilon_k}}\}$ denote the $k$-dimensional
probability distribution of a random vector $(X_1, \dots, X_k) \in
\{0,1\}^k$. Then
\[\frac{1}{\log 2}D(\{p^{\bm{\epsilon_k}}\} \|
\{\rnp^{\bm{\epsilon_k}}\})
=
\sum_{\bm{\epsilon_k} \in \{0,1\}^k} p^{\bm{\epsilon_k}}
\log_2 \frac{p^{\bm{\epsilon_k}}}{2^{-k}}\\
=k  - H(X_1,\dots, X_k),
\]
where 
$H(X_1,\dots, X_k)$
denotes the entropy of 
$\{p^{\bm{\epsilon_k}}\}$.  For an infinite sequence $X_1,X_2, \dots$ 
of stationary and ergodic 0-1 random variables, the entropy 
$H({\cal X})=H(X_1,X_2, \dots)$ is defined as 
$H({\cal X}) = \lim_k (1/k) H(X_1, \dots, H_k)$.  Under the
assumption of stationarity and ergodicity, each 
$k$-dimensional empirical distribution converges to 
the probability distribution $\{p^{\bm{\epsilon_k}}\}$
almost surely.  
Furthermore  in the previous
subsection we saw that larger block sizes achieve better growth rates.
Therefore, arguing as in Chapter 13 of \cite{cover-thomas}
we have the following proposition.

\begin{proposition} 
\label{prop:1}
If $X_1, X_2, \dots$, are stationary and ergodic sequence of 
0-1 random variables, then
\[
\frac{1}{n}\log_2 \cK_n^{\cP^*} \rightarrow 1 - H({\cal X}), \
\ a.s.\ \ 
\qquad
(n\rightarrow \infty)  .
\]
\end{proposition}

Thus we can say that  $\cK_n^{\cP^*} \to \infty$
achieves the optimal rate in the sense of universal source coding in
information theory.

\section{Application to asset trading games in continuous time}
\label{sec:preliminary asset}

In this section we apply the results on 
the block strategy of length two  and the first-order Markovian
strategy of the previous section to 
an asset trading game in continuous time considered 
in \cite{tkt1}. It should be noted that 
our interest here is to derive explicit growth rates of our strategies
applied to the asset trading game, rather than a rigorous treatment of 
forcing of the variation exponent of two.    Therefore in our derivation we proceed with
informal definitions and convenient regularity conditions.

In Section \ref{subsec:informal} we summarize the setup of 
an asset trading game.
In Section \ref{subsec:simple} we obtain growth rates of our
strategies for the asset trading game when the 
variation exponent of the asset price path deviates from two.

\subsection{Preliminaries on asset trading games in continuous time}
\label{subsec:informal}

Here we summarize preliminary facts on the asset trading game
as formulated in \cite{tkt1}.  Our framework in \cite{tkt1} is now much
generalized in the recent papers of Vovk
(\cite{vovk-cont-1},\cite{vovk-cont-2},\cite{vovk-cont-3}).

Suppose that there is a financial asset which is traded in
continuous time.  Let $S(t)$ denote the price of the unit amount of
the asset at time $t$.  We assume that $S(t)$ is positive and a
continuous function of $t$.  The price path $S(\cdot)$ is
chosen by a player ``Market'', which is the same as Reality in the
coin-tossing game.  ``Investor'' enters the market at time
$t=t_0=0$ with the initial capital of ${\cal K}(0)=1$.  
He decides 
discrete time points $0=t_0 <  t_1 < t_2 < \cdots$ to  trade  the
financial asset. The
trading time $t_i$ and the amount 
$M_i$ of the asset Investor holds
for the interval $[t_i,t_{i+1})$ can depend 
on the path of $S(t)$ up to time $t_i$.  

The basic fact on the behavior of $S(t)$ is the ``$\sqrt{{\rm
    d}t}$-effect'' (\cite{wp5}, \cite{tkt1}), 
which asserts that infinitesimal increments
$|{\rm d}S(t)|$ of the price path have to be of the
order $O(\sqrt{{\rm d}t})$,  in the sense that otherwise Investor can
make arbitrarily large profit without risking bankruptcy.  
When $|{\rm d}S(t)|=O(({\rm d}t)^H)$, then $H \in (0,1]$ is called the
H\"older exponent or the Hurst index of $S(t)$, and $1/H$ is called the
variation exponent.
Thus the game-theoretic statement of  $\sqrt{{\rm     d}t}$-effect is that
Investor can force the variation exponent of two.

We  consider ``limit order'' 
strategy of Investor.  Let $\delta>0$ be a constant.
Investor determines the trading times  $t_1, t_2, \dots$ as follows.
After $t_i$ is determined,
let $t_{i+1}$ be the first time after $t_i$ when either
\begin{align}
\label{eq:2-3}
\frac{S(t_{i+1})}{S(t_i)} = 1 + \delta\quad \textrm{or}
\quad = \frac{1}{1 + \delta}
\end{align}
happens. This procedure leads to a discrete time coin-tossing game
embedded into the asset trading game as follows.  Let
\[
x_n = \frac{(1+\delta) S(t_{n+1}) - S(t_n)}
{\delta (2+ \delta) S(t_n)}
= \begin{cases}
 1, & \text{if}\  S(t_{n+1})=S(t_n) (1+\delta)\\
 0, & \text{if}\  S(t_{n+1})=S(t_n)/(1+\delta),
\end{cases}
\]
and \[
\rnp = \rnp_\delta = \frac{1}{2+ \delta}.
\]  
Also write
$
\tilde{{\cal K}}_n = {\cal K}(t_{n+1})$.
Clearly  $x_n$ can be thought as ``heads'' or ``tails'' chosen by Market
after the time $t_i$.  More formally we define 
the following protocol of an embedded 
discrete time coin-tossing game.

\medskip\noindent
\textsc{Embedded Discrete Time Coin-Tossing Game} \\
\textbf{Protocol:}

\parshape=6
\IndentI   \WidthI
\IndentI   \WidthI
\IndentII  \WidthII
\IndentII  \WidthII
\IndentII  \WidthII
\IndentI   \WidthI
\noindent
$\tilde{{\cal K}}_0 :=1$.\\
FOR  $n=1, 2, \dots$:\\
  Investor announces $\nu_n\in{\mathbb R}$.\\
  Market announces $x_n\in \{0, 1\}$.\\
  $\tilde{{\cal K}}_n = \tilde{{\cal K}}_{n-1}(1 + \nu_n
  (x_n-\rnp))$.\\
END FOR\\

This embedded discrete time game allows us to apply results on
coin-tossing games to the asset trading game in continuous time.  
The amount $M_i$ of the asset held during $[t_i, t_{i+1})$ is determined by our
strategies in Section \ref{sec:patterns}. 
{}From now on we  fix the time interval $[0,T]$ of the asset trading
game.   Then 
the total number rounds  $n=n(\delta)$
played in the embedded coin-tossing game is finite.
For a given path, $n(\delta)$ is increased by 
letting  $\delta \downarrow 0$ in (\ref{eq:2-3}).  
$n(\delta)$ diverges to infinity as $\delta\downarrow 0$, unless $S(t)$ is
constant on $[0,T]$. 
We call a strategy with small $\delta$ a high-frequency strategy.  
In \cite{tkt1} we
applied high-frequency Bayesian strategy of \cite{ktt1} to the embedded discrete
time game and proved that 
Investor can force the variation exponent of
two, within an arbitrary small constant.

However the Bayesian strategy of \cite{ktt1} is based only on $s_n$
and does not take the higher order patterns of the increments of
$S(t)$ into account.  In the previous section we saw that multistep
Bayesian strategies can effectively exploit higher order patterns when
Reality's moves are not random.
Therefore it is of interest to investigate how fast Investor can increase his
capital by a high-frequency multistep Bayesian strategy, if the
variation exponent of $S(t)$ deviates from two.

\subsection{Growth rates of block strategy of length two and first-order Markov strategy
for asset trading game}
\label{subsec:simple}

In this section we derive growth rates of high-frequency
block strategy of length two and first-order Markov strategy
in the embedded coin-tossing game.
Our results are stated in two propositions at the end of this section.

Write $\eta=\log(1+\delta)$.  Then $\eta\downarrow 0$ is equivalent to
$\delta\downarrow 0$. We decrease $\eta$  to zero as
\[
\eta_k =  2^{-k},\ k = 1, 2, \dots.
\]
The advantage of taking this sequence of $\eta_k$ is that the
equi-spaced grids for $\log S(t)$  are completely nested in $k$ and we can establish
some important relations between the empirical distributions of block
patterns for different $k$ (see Lemma \ref{lem:combinatorial}  below).
We call the embedded coin-tossing game with $\eta_k$ the $k$-th
embedded coin-tossing game.
Let $n(\eta_k)$ denote the total number of rounds of 
the $k$-th embedded coin-tossing game. 
For notational convenience we sometimes write $n$ or $n_k$  
instead of $n(\eta_k)$.
As in Section \ref{subsec:block}, let $m_{n_k/2}^{ij}$, $i,j=0,1$,
denote the number of pairs $(ij)$ among $(x_1x_2)(x_3x_4)\dots$
in the $k$-th embedded coin-tossing game.  
For the shift of one, let
$\tilde m_{n_k/2}^{ij}$ 
denote the number of pairs $(ij)$ among $(x_2x_3)(x_4x_5)\dots$
in the $k$-th embedded coin-tossing game. 
For these counts of pairs, it is more precise to write
$m_{[n_k/2]}^{ij}$ or $\tilde m_{[(n_k-1)/2]}^{ij}$. 
However for simplicity we write $m_{n_k/2}^{ij}$ or 
$\tilde m_{n_k/2}^{ij}$  in the following. We define 
$q_{n_k}^i, q_{n_k}^{ij}$ as in Section 
\ref{subsec:Markov}.

We now give a preliminary consideration on the behavior of counts
$n(\eta_k)$, $q_{n_k}^i$, $q_{n_k}^{ij}$, 
$m_{n_k/2}^{ij}$, $\tilde m_{n_k/2}^{ij}$ for different $k$.
At this point it is helpful to consider properties of fractional Brownian motion
(Chapter 4 of \cite{embrechts-maejima}).
Let $\{B_H(t)\}$  denote the fractional Brownian motion of Hurst index
$H$.  $B_H(t)$ corresponds to $\log S(t)$ in the asset trading game.
$B_H(t)$ is a typical stochastic process with $|{\rm d}B_H(t)|=O(({\rm d}t)^H)$.
$\{B_H(t)\}$ is self-similar, i.e., for every $a>0$ the distribution
of  $\{B_H(at)\}$ coincides with that of $\{a^H  B_H(t)\}$.
This implies that making the grid finer as  $\eta_k \rightarrow
\eta_{k}/2$ is equivalent (in distribution) to increasing $T$ as $T
\rightarrow 2^{1/H} T$.
This suggests that when Market chooses a path $S(t)$ with a fixed exponent
$H$, then 
\begin{equation}
\label{eq:convenient-ass-1}
n_{k+1} \simeq 2^{1/H} n_k.
\end{equation}
Furthermore $\{B_H(t)\}$ has stationary increments, i.e., the
distribution of the increments of $\{B_H(t)\}$ are invariant with
respect to arbitrary time shift.  This corresponds to our assumption of 
homogeneity with respect to the shifts in Section
\ref{sec:comparison-of-rate}.    Under the homogeneity assumption we expect
\begin{equation}
\label{eq:convenient-ass-2}
2 m_{n_k/2}^{ij} \simeq 2 \tilde m_{n_k/2}^{ij} \simeq q_{n_k}^{ij}.
\end{equation}
Also note the following trivial combinatorial relations for any $n$:
\[
q_n^{11}+ q_n^{01} = q_n^1 - x_1, \ 
q_n^{11}+ q_n^{10} = q_n^1 - x_n \quad \text{and}\quad 
|q_n^{01} - q_n^{10}|  = |x_n - x_1| \le 1.
\]

{}From game-theoretic viewpoint Investor can force $H=1/2$.  If
Market chooses a path with $H \neq 1/2$, then the notion of forcing can not
be applied and there is no guarantee that 
(\ref{eq:convenient-ass-1}) and (\ref{eq:convenient-ass-2}) hold.  
However even for $H\neq 1/2$ we use 
(\ref{eq:convenient-ass-1}), (\ref{eq:convenient-ass-2}) as
convenient regularity conditions for evaluating the growth rates of
our strategies in view of the properties of the
fractional Brownian motion.  In Proposition \ref{prop:markov-asset}
and Proposition \ref{prop:block-asset}
below we obtain the growth rate of the first-order Markovian strategy 
and block strategy of length two under these conditions.
In these propositions the approximate
equalities in (\ref{eq:convenient-ass-1}) and (\ref{eq:convenient-ass-2}) are 
understood in the sense that the ratios two sides converge to 1.

Now we state  the following crucial combinatorial fact.
\begin{lemma}  
\label{lem:combinatorial}
For each  path $S(t)$ and for each  $k$ 
\begin{equation*}
m_{n_k/2}^{11}=q_{n_{k-1}}^1, \qquad
m_{n_k/2}^{00}=q_{n_{k-1}}^0.
\end{equation*}
\end{lemma}

\begin{proof}
Consider two nested equi-spaced grids with intervals  $\eta_{k-1}$ and $\eta_k=\eta_{k-1}/2$.
For the grid with the interval
$\eta_{k-1}$, consider a step, where the price is going upward  from $t_i$
to $t_{i+1}$ i.e.\ $\log S(t_{i+1}) = \log S(t_i) + \eta_{k-1}$ in (\ref{eq:2-3}). 
It is obvious that this upward step corresponds exactly to two
consecutive upward steps for the the grid with the interval
$\eta_k$. Therefore
$m_{n_k/2}^{11}=q_{n_{k-1}}^1$.  By counting downward steps, we similarly
obtain
$m_{n_k/2}^{00}=q_{n_{k-1}}^0$.  
\end{proof}

We now derive the growth rate of the first-order Markovian strategy.
Define
\begin{align*}
TV(\eta_k, T) &= \sum_{i=1}^{n_k} |\log S(t_i) - \log S(t_{i-1})| = 
n_k\eta_k = (q_{n_k}^1 + q_{n_k}^0)\eta_k,
\\
L(\eta_k, T) &= \log S(t_{n_k}) - \log S(0) = (q_{n_k}^1 - q_{n_k}^0)\eta_k.
\end{align*}
Under (\ref{eq:convenient-ass-1}),  $n_{k+1} \eta_{k+1} \simeq 2^{1/H-1}n_k \eta_k$.
Therefore  $n_{k+1} \eta_{k+1}\rightarrow\infty$ as
$k\rightarrow\infty$  for the case $H<1$.  On the other hand $S(t_{n_k})
\rightarrow S(T)$ as $k\rightarrow\infty$ and 
$L(\eta_k, T) \rightarrow \log S(T)-\log S(0)$.  Therefore for each
path, 
$L(\eta_k, T)/TV(\eta_k, T) \rightarrow 0$ as  $k\rightarrow\infty$
and this implies that 
\begin{equation}
\label{eq:1/2-limit}
\frac{q_{n_k}^1}{n_k} \rightarrow \frac{1}{2}  \qquad (k \rightarrow \infty).
\end{equation}
Also note that $1/2=\lim_{\delta\rightarrow 0}\rnp_\delta$.

Furthermore by Lemma \ref{lem:combinatorial}, under (\ref{eq:convenient-ass-1}) 
\[
q_{n_k}^{11} \simeq 2 m_{n_k/2}^{11} = 2 q_{n_{k-1}}^1 \simeq n_{k-1}.
\]
Therefore
\[
r_{n_k}^1 = \frac{q_{n_k}^{11}}{q_{n_k}^1} \simeq \frac{
  n_{k-1}}{n_k/2} \simeq \frac{1}{2^{1/H-1}}.
\]
Similarly  $r_{n_k}^0 =q_{n_k}^{01}/q_{n_k}^0 \simeq 1 - 1/2^{1/H-1}$.
Let $\cK_{n_k}^M$ denote the capital of the first-order Markovian strategy
at the end of the $k$-th embedded coin-tossing game.
By (\ref{eq:markov2}) we obtain  the following proposition.

\begin{proposition} 
\label{prop:markov-asset}
Suppose that Market  chooses a path such that 
\[
1=\lim_{k\rightarrow\infty} \frac{n_{k+1}}{2^{1/H}n_k}
=\lim_{k\rightarrow\infty}\frac{2 m_{n_k/2}^{ij}}{q_{n_k}^{ij}} 
=\lim_{k\rightarrow\infty}\frac{2 \tilde m_{n_k/2}^{ij}}{q_{n_k}^{ij}}, \quad
i,j=0,1.
\]
Then 
\[
\lim_{k\rightarrow\infty} \frac{1}{n_k} \log \cK_{n_k}^M=
 D\Big(\frac{1}{2^{1/H-1}} \Big\| \frac{1}{2} \Big) .
\]
\end{proposition}

We now consider the block strategy of length two.
Let $\cK_{n_k}^B$  denote the capital of block strategy at the end of
the $k$-th embedded coin-tossing game.
By (\ref{eq:block-markov-average}) we know that
$\log \cK_{n_k}^B$ is the average of 
$\log \cK_{n_k}^{M,1}$ and $\log \cK_{n_k}^{M,0}$. However
by (\ref{eq:1/2-limit}) the growth rate of $\cK_{n_k}^{M,0}$ is zero.
Therefore we have the following result.

\begin{proposition} 
\label{prop:block-asset}
Under the same assumption as in Proposition \ref{prop:markov-asset}
\[
\lim_{k\rightarrow\infty} \frac{1}{n_k} \log \cK_{n_k}^B=
 \frac{1}{2}D\Big(\frac{1}{2^{1/H-1}} \Big\| \frac{1}{2} \Big) .
\]
\end{proposition}

Therefore the growth rate of the block strategy is half of rate of
the Markovian strategy.  $\cK_{n_k}^{M,0}$ is the capital process of
the strategy based only on the past average of Reality's moves
considered in \cite{tkt1}, whose growth rate is zero.  It is of
interest to note that despite this zero growth rate, 
the strategy in \cite{tkt1} was sufficient to force the variation
exponent of two of the Market's path.  
This suggests that looking for a simple strategy for forcing certain
event and looking for a more aggressive strategy with a better growth
rate need different considerations.

\section{Discussions}
\label{sec:discussions}

In this paper we studied multistep Bayesian strategies for
coin-tossing games.  Our general conclusion was that asymptotically we
obtain better growth rates by incorporating larger block
sizes for block strategies or longer orders for Markovian strategies. 
However this conclusion has to be taken with the following cautions.
When the main growth term expressed in terms of the Kullback
divergence is close to zero, we have to compare this to the term
of order $O(\log n)$.  Generally the term of order $O(\log n)$
can be understood as a penalty term for larger models, i.e., for
using strategies incorporating larger blocks. Therefore if the
coin-tossing game is played only a finite number of rounds, or Reality
does not deviate too much from the independent Bernoulli trials, then
it might be advantageous to use shorter block sizes.  This is
essentially the same tradeoff as in  statistical model selection
based on various information criteria.    It is of great interest to
consider selecting among strategies or dynamically
adjusting weights for them.

For convenience we made the assumption of homogeneity for block
strategies in Section \ref{sec:comparison-of-rate} and in Section
\ref{sec:preliminary asset}.  We initially thought that 
homogeneity can be ``forced'' on Reality by appropriate strategies of
Skeptic.  However, when Reality deviates from 
independent Bernoulli trials, the game-theoretic notion of
forcing can not be applied.
Intuitively it seems that Skeptic
can further exploit patterns in Reality's moves when the homogeneity
with respect to shifts does not hold.  However at present it seems 
difficult to formulate results in this direction.

In Section \ref{subsec:universal} we considered an infinite countable
mixture of block strategies and Markovian strategies.  Using this
countable mixture, Skeptic can asymptotically exploit any
deviation of Reality's moves from 
independent
Bernoulli trials.  We 
pointed out that the  idea of the universal source
coding in information theory is similar.  Our result is also very closely
connected to results in algorithmic theory of randomness.  
We can think of each component strategy as a test of randomness of
Reality's moves.
In algorithmic randomness
there are strong computability restrictions on the allowed sample spaces.
In the game-theoretic approach we do not have to worry about
computability and by appropriate discretization it is now possible to 
discuss the randomness of continuous paths.

In Section \ref{subsec:simple}  we only considered
block strategies of length two and first-order Markovian strategies
in the embedded coin-tossing game.   We could obtain
the explicit descriptions for the
growth rates in  Proposition \ref{prop:markov-asset}
and Proposition \ref{prop:block-asset} because of the combinatorial
fact of Lemma \ref{lem:combinatorial}.  It is of interest to
investigate growth rates of higher-order Markovian strategies in the
asset trading game.

For measure-theoretic stochastic processes, the regularity conditions
assumed in
Propositions \ref{prop:markov-asset} and \ref{prop:block-asset} are
basically law of large numbers, and we expect that they hold for
fractional Brownian motions.  However the trading times in
(\ref{eq:2-3})  are stopping
times  and  the fractional Brownian motion for  $H\neq 1/2$ is not a
Markov process.  Therefore it is  not easy to prove that
the regularity conditions hold for  fractional Brownian motions. 

\appendix
\section{Equivalence of Bayesian strategy and prudent strategy 
in coin-tossing games}

Here we establish a one-to-one correspondence between Skeptic's prudent
strategy and  a probability distribution on the set of paths $\Xi$ in the
coin-tossing game.

For one direction suppose that Skeptic models Reality's moves by a
probability distribution $Q$.  Write $\alpha_n = M_n / \cK_{n-1}$.
Given $\cK_{n-1}$ assume that Skeptic tries to maximize the
conditional expected value of $\log \cK_n$.  It is equivalent to
maximizing
\begin{equation}
\label{eq:maximize-log}
 p_n \log (1+\alpha_n(1-\rnp)) + (1-p_n) \log (1-\alpha_n \rnp) 
\end{equation}
with respect to $\alpha_n$,  
where  $p_n = p^Q_n$ is given in (\ref{eq:conditional-probability}).
The maximizing value of $\alpha_n$ is uniquely given as
\[
\alpha_n = \frac{p_n - \rnp}{\rnp(1-\rnp)}.
\]
With this $\alpha_n$,
\[
\cK_n = 
\begin{cases} \cK_{n-1} p_n/\rnp , &  \textrm{if}\ x_n=1 \\
              \cK_{n-1} (1-p_n)/(1-\rnp), & \textrm{if}\ x_n=0.
\end{cases}
\]
Note that $\cK_n=0$ if either $p_n=0$ and $x_n=1$ or $p_n=1$ and
$x_n=0$.  In this case Skeptic can not play any more.  For other cases
he can keep playing the game. It should be noted that this is
consistent with the definition of conditional probability in
(\ref{eq:conditional-probability}), namely, Skeptic can continue the
game if and only if (\ref{eq:conditional-probability}) is defined. We
have shown that a probability distribution $Q$ leads to the strategy
given in (\ref{eq:QtoP}).

For another direction let $\cP$ be  a  prudent strategy of Skeptic.
Starting with the initial capital of $\cK_0=1$,
define
\begin{align*}
Q_1(1)&=\rnp + M_1 \rnp(1-\rnp)=\rnp(1 + M_1 (1-\rnp)), \\
Q_1(0)& = 1-\rnp - M_1 \rnp(1-\rnp)=(1-\rnp) (1- M_1 \rnp).
\end{align*}
Then $Q_1(0)$ and $Q_1(1)$ are non-negative and  $1 = Q_1(0)+
Q_1(1)$.  For the case  
$\cK_{n-1}(\xi^{n-1})>0$
recursively define
\begin{align*}
Q_n(\xi^{n-1}1) & = \rnp Q_{n-1}(\xi^{n-1})
\Big( 1 +  \frac{M_n(\xi^{n-1})}{\cK_{n-1}(\xi^{n-1})}(1-\rnp) \Big),\\
Q_n(\xi^{n-1}0) & = (1-\rnp) Q_{n-1}(\xi^{n-1}) 
\Big( 1- \frac{M_n(\xi^{n-1})}{\cK_{n-1}(\xi^{n-1})}\rnp \Big).
\end{align*}
These are non-negative and satisfy the consistency condition 
(\ref{eq:consistency}). If $\cK_{n-1}(\xi^{n-1})=0$, then
define $0=Q_n(\xi^{n-1}1)=Q_n(\xi^{n-1}0)$, which is
also consistent. By this procedure a Skeptic's prudent strategy leads to a
probability distribution 
$\cP \mapsto Q$.  

By construction it is
obvious  that this map is the inverse map to
(\ref{eq:QtoP})
and  therefore there exists a one-to-one  correspondence between the set of probability distributions and the set of Skeptic's strategies
satisfying the collateral duty.

Finally we state the following Bayesian optimality result, which
follows easily from the maximization in (\ref{eq:maximize-log})

\begin{proposition}  Let $Q$ be a probability distribution on $\Xi$
  and let $\cP$ be the strategy corresponding to $Q$. For any other
  strategy $\tilde \cP$
\[
E^Q (\log \cK_n^{\cP_Q}) \ge E^Q (\log \cK_n^{\tilde \cP}).
\]
\end{proposition}

\end{document}